\documentclass[11pt,twoside]{article}
\usepackage[letterpaper,margin=1in]{geometry}
\usepackage[utf8x]{inputenc}
\usepackage{amssymb}
\usepackage{graphicx}
\usepackage{verbatim}
\usepackage{hyperref}

\usepackage{amsthm,amsmath}

\usepackage{dsfont}
\usepackage{xargs}
\usepackage{todonotes}
\usepackage{tikz}
\usepackage[boxruled,linesnumbered,lined]{algorithm2e}
\SetKwRepeat{Do}{do}{while}

\usepackage{amssymb}
\usepackage{wrapfig}
\usepackage{mathrsfs}
\usepackage{enumerate}
\usepackage{caption}
\usepackage{microtype}
\usepackage{soul}
\usetikzlibrary{decorations.pathreplacing}

\newtheorem{theorem}{Theorem}
\newtheorem{lemma}[theorem]{Lemma}

\newtheorem{definition}[theorem]{Definition}

\newtheorem{observation}[theorem]{Observation}

\newtheorem{remark}[]{Remark}

\newcommand{\TT}{\mathcal{T}}
\newcommand{\LL}[1]{\log^{(#1)} n}
\newcommand{\EL}{\ell}
\newcommand{\NN}{\mathcal{N}}
\newcommand{\HH}{\mathcal{H}}
\newcommand{\SSS}{\mathcal{S}}

\newcommand{\DD}{\text{deg}}
\newcommand{\CC}{\textsc{Color}}
\newcommand{\DFS}{\textsc{Dfs}}
\newcommand{\PU}{\textsc{push}}
\newcommand{\PO}{\textsc{pop}}

\newcommand{\GR}{\textsc{gray}}
\newcommand{\WH}{\textsc{white}}
\newcommand{\BL}{\textsc{black}}
\newcommand{\REE}{\textsc{Restore-Empty}}
\newcommand{\REF}{\textsc{Restore-Full}}
\newcommand{\GRO}{\textsc{group}}
\newcommand{\IL}[1]{#1\text{-light}}
\newcommand{\IH}[1]{#1\text{-heavy}}
\newcommand{\IS}[1]{#1\text{-segment}}
\newcommand{\IN}{\textsc{Insert}}
\newcommand{\DL}{\textsc{Delete}}

\newcommand{\GET}{\textsc{Search}}
\newcommand{\SEG}{\textsc{Seg}}
\newcommand{\NXT}{\textsc{Next}}

\newcommand{\UU}{\mathcal{U}}
\newcommand{\ME}{\textsc{Member}}

\title{Nearly Optimal Space Efficient Algorithm for Depth First Search}

\author{Jayesh Choudhari$^1$ \and Manoj Gupta$^1$ \and Shivdutt Sharma$^1$}

\date{
	$^1$IIT Gandhinagar, India\\%
}
\begin{document}

\maketitle
\begin{abstract}
We design a space-efficient algorithm for performing depth-first search traversal($\DFS$) of a graph in $O(m+n\log^* n)$
time using $O(n)$ bits of space. While a  normal
$\DFS$ algorithm results in a $\DFS$-tree (in case the graph
is connected), our space bounds do not permit us even to store such a tree. However, our algorithm correctly outputs all edges of the $\DFS$-tree.

 The previous best algorithm (which used $O(n)$ working space) took $O(m \log n)$ time (Asano, Izumi, Kiyomi, Konagaya, Ono, Otachi, Schweitzer, Tarui, Uehara (ISAAC 2014) and 
Elmasry, Hagerup, Krammer (STACS
2015)). The main open question left behind in this area was to design faster algorithm for $\DFS$ using $O(n)$ bits of space. Our algorithm answers this open question as it has a nearly optimal  running time\ (as the $\DFS$ takes $O(m+n)$ time even if there is no space restriction).    
\end{abstract}
 \section{Introduction}
In analyzing algorithms, mostly we concentrate on minimizing
the running time, or the quality of the solution (if the
problem is {\em hard}). After we have optimized the above parameters, we then look to reduce the space taken by the algorithm, if possible. An excellent theoretical question is: {\em Given a problem $P$, design an algorithm that solves it in as low space as possible}. These algorithms are called space-efficient algorithms as we want to optimize on the space taken by the algorithm while not increasing the running time by much (compared to the best algorithm for the problem with no space restriction).

Recently, designing space-efficient algorithms has gained importance because of the rapid growth in the use of mobile devices and other hand-held devices which come with limited memory (e.g., the devices like Raspberry Pi, which are widely used in IoT applications). Another crucial reason for the increasing importance of the space-efficient algorithms is the rate and the volume at which huge datasets are generated (``big data"). Areas like machine learning, scientific computing, network traffic monitoring, Internet search, signal processing, etc., need to process big data using as less memory as possible.   

Algorithmic fields like Dynamic Graph Algorithm \cite{BaswanaKS12,HolmLT01,RodittyZ08,RodittyZ12,Thorup07}
and Streaming algorithm \cite{ahn2012analyzing,ahn2012graph,bhattacharya2015space,sarma2011estimating,ahn2012graph,ahn2013linear,ahn2012analyzing} mandate low space usage by the algorithm. In a streaming
algorithm, the mandate is mentioned upfront. In
a dynamic graph algorithm, this mandate is implied as we want
the {\em update time} of the algorithm to be as low
as possible. Low update time implies that we don't have
enough time to look at our data-structure. Thus, we want
our data-structure to be as compact as possible.
Motivated by the growing body of work in the field of space-efficient algorithms, this paper focuses on optimizing the space taken by the DFS algorithm, which is one of the fundamental graph algorithms.

However, one needs to be slightly cautious about the definition of space. For a graph problem, it would take $O(m+n)$ space
just to represent the graph. So, it seems that any graph problem requires $\Omega(m+n)$ bits. To avoid such trivial answers, we first define our model of computation.

\subsection{Model of Computation : Register Input Model
\cite{Frederickson87}}
Frederickson \cite{Frederickson87} introduced the  \emph{register input model}  in which  the input (graph -- in this case) is given in a read-only memory (thus, it cannot be modified). Also the output of the algorithm  is written on a write-only memory. Along with the input and the output memory, a random-access memory of limited size is also available. Similar to the standard RAM model, the data on the input memory and the workspace is divided into words of size $\Theta(\log n)$ bits. Any arithmetic, logical and bitwise operations on constant number of words take $O(1)$ time.

When we say that our algorithm uses $O(n)$ bits, this is the space on the random-access memory used by our algorithm. The above model takes care of the case
when the input itself takes a lot of space --- by designating a special read-only memory for the input.

We highlight some  results that make use of the register input model. Pagter and Rauhe \cite{PagterR98} described a comparison-based algorithm for sorting $n$ numbers: for every given $s$ with $\log n \le s \le n/\log n $, an algorithm that takes  $O(n^2/s)$ time using $O(s)$ bits. A matching lower bound of $\Omega(n^2)$ for the time-space product was given by Beame \cite{Beame91} for the strong branching-program model. Please see references for other problems in this model \cite{Chan10,ChanMR13,ElmasryJKS14,MunroR96,RamanR99,AsanoBBKMR14,AsanoMRW11,BarbaKLS14,BarbaKLSR15,DarwishE14}. In this paper, our main focus is on the Depth First Search Problem.

\subsection{DFS Problem }
The problem of space efficient $\DFS$ has received a lot
of attention recently. Asano et al. \cite{Asano2014}
designed an algorithm that can perform  $\DFS$ in (unspecified) {\em polynomial
time} using $n + o(n)$ bits. If the space is increased to
$2n + o(n)$ bits then their running time decreases to $O(mn)$.
They also showed how to perform $\DFS$ in $O(m \log n)$
time using $O(n)$ bits. Elmasry et al. \cite{ElmasryHK15}
improved this result by designing an algorithm
that can perform $\DFS$ in $O(m+n)$ time using $O(n
\log \log n)$ bits. Banerjee et al.\cite{BanerjeeC016} proposed an efficient $\DFS$ algorithm that takes $O(m+n)$ time using $O(m+n)$
space. Note that  this is a strict improvement (over the Elmasry et al. \cite{ElmasryHK15}
 result) only if the
graph is sparse. The following open question was raised by Asano et al. \cite{Asano2014}
 in their paper: 
\begin{center} 
 {\em Using $O(n)$ space, can $\DFS$ be done in $o(m \log n)$ time?}
\end{center}

Recently, Hagerup \cite{Hagerup18} claimed an algorithm that finds $\DFS$ in $O(m \log^* n+n)$ time using $O(n)$ bits of space. We improve upon this algorithm giving a near optimal running time for $\DFS$ --- it is almost linear in $m+n$. Our result can be succinctly stated as follows: 

\begin{theorem}
\label{thm:main}
There exists a randomized algorithm that can perform $\DFS$
of a given graph in $O(m+n\log^* n)$ time with a high probability
($(1-1/n^c)$ (where $c \ge 1$)) using $O(n)$ bits of space.
(Note that our algorithm is randomized because we
use  succinct dictionaries that use random bits)\end{theorem}

The succinct  dictionary (used by our algorithm)  performs insertion/deletion in $O(1)$ time with a
probability of $(1-1/n^c)$ (where $c \ge 3$). Our algorithm performs  at most $ O(n \log^* n+m)$ insertions/deletions across all  dictionaries. Hence, the probability that our algorithm takes more than $O(1)$ time for any of these $ O(n \log^* n+m)$ insertions/deletions is $O(1/n^{c-2})$ (by union bound). 

 \section{Overview}
\label{sec:overview}

We will assume that vertices of input graph $G$ are numbered from 1 to $n$. Let $\NN(v)$ denote the
neighborhood of the vertex $v$ and $\NN(v)[k]$ denote the $k$-th neighbor
of the vertex $v$, where $1 \le k \le |\NN(v)|$. As  in   \cite{ElmasryHK15},
we will assume that $\NN(v)$ is an array. So, we have random
access to any element in this array. Also, we implicitly
know the degree of $v$, $\DD(v) = |\NN(v)|$.

Normally, the $\DFS$ algorithm outputs the $\DFS$ tree. Given the space bounds, we cannot store the $\DFS$ tree, but, we output the edges of the $\DFS$ tree as soon as we encounter them.
We view that the problem
is solved if  the output edges form a  valid $\DFS$ tree. 

  We first give a quick overview of the non-recursive implementation
of the $\DFS$ algorithm. Let $G(V,E)$ be the input graph having $n$ vertices and $m$ edges.
For this implementation, we will use a stack $\SSS$.
Initially, all  vertices are colored white and assume that we
start the $\DFS$ from a vertex $u$.
So, $u$ is added to the stack $\SSS$. The algorithm then processes all  elements of the stack till it becomes empty. Thus, the top vertex, say $u$, is popped from the stack and is {\em processed} as follows: each neighbor of $u$ is explored. If a white vertex $v$ is found, then $u$ is pushed on to the stack and {\em processing} of $v$ starts. If none of the neighbors of vertex $u$ are white, then $u$ is colored black.
 Whenever $u$ discovers a white vertex
$v$, we push a tuple $(u,u.\NXT)$ on to $\SSS$, where
the second entry in the tuple tells us which neighbor of vertex $u$ to explore
once processing of $u$ resumes.

Now, let us formally define the second entry in the tuple $(u,u.\NXT)$

\begin{definition}
For any vertex $u$, if $(u,u.\NXT)$ is an entry on the stack $\SSS$, then $\NN(u)[u.\NXT]$ denotes the first neighbor of the vertex $u$ which is still not explored while processing $u$.

\end{definition}

\begin{algorithm}
  \caption{ \textsc{Initialize}()}
  \label{alg:imag}

  \For{$ i \leftarrow 1$ to n}
  {
      $\CC(i) \leftarrow \WH$\;
  }

  \ForEach{$u \in V$}
  {
      \If{$u$ is $\WH$}
      {
           \textsc{Process}($u)$\;
      }
  }
\end{algorithm}

\begin{algorithm}[h]

  \caption{ \textsc{Process}$(u)$}
  \label{}
  $\SSS.\textsc{push}(u,1)$\;

   \While{$\SSS$ is not empty }
  {

     $(v,k) \leftarrow $ top element of $\SSS$\;

     $\CC(v) \leftarrow \GR$\;
     \tcc{scan neighbors of $v$}
     \If{$k \le \DD(v)$}
     {
              $\SSS.\textsc{push}(v,k+1)$

         \If{ $\CC(\NN(v)[k])$ is \WH }
         {
                 output edge $(v,\NN(v)[k])$\;

                 $\SSS.\textsc{push}(\NN(v)[k],1)$\;

         }

     }
     \Else
     {
         $\CC(v) \leftarrow \BL$ \;
     }

  }
\end{algorithm}

The space
required to represent the first and second term of each tuple
in the stack  $\SSS$ is
$O(\log
n)$ bits. As there are $n$  vertices in the graph, the
size of the stack can reach $\Omega(n)$ in the worst case.
So, the
total space taken by the trivial algorithm is $O(n \log
n)$ bits.

Our algorithm closely follows \cite{ElmasryHK15}. So, we first give a brief overview
of their approach and later, we will explain our improvement over their approach.

\subsection{Previous Approach (Elmasry et. al. \cite{ElmasryHK15})}
The trivial $\DFS$ algorithm does not work for Elmasry et al.\cite{ElmasryHK15} because the stack $\SSS$ itself takes  $O(n \log n)$ bits of space.
Hence, stack $\SSS$ is not implemented --- but, is referred to as an {\em imaginary} stack.
Let the stack $\SSS$ be  divided into segments of size $\frac{n}{\log n}$ --- the first segment is the bottommost $\frac{n}{\log n}$ vertices of $\SSS$, the second segment is the next $\frac{n}{\log n}$ vertices of $\SSS$ and so on. A new stack $\SSS_1$ is implemented, which contains vertices from at most top two segments of the imaginary stack $\SSS$. Each entry of the stack $\SSS_1$ is a tuple: $(v,v.\NXT)$ where $v \in V$. The space required to represent these two terms is at most $2\log n$.  Thus, the total space required for $\SSS_1$ is $O(\frac{n}{\log n} \times \log n) = O(n)$ bits. Since,
the size of $\SSS_1$ is very small as compared to the imaginary
stack $\SSS$, the main problem arises when an element
is to be pushed on $\SSS_1$ but it is full or when $\SSS_1$ becomes empty (but $\SSS$ contains vertices). Thus, there is a need to make space in $\SSS_1$ or a way to {\em restore} vertices in $\SSS_1$.

To handle the case when $\SSS_1$ is full, Elmasry et al.\cite{ElmasryHK15} remove the bottom half elements of $\SSS_1$. So, a new entry can now be pushed on to $\SSS_1$, and the $\DFS$ algorithm can proceed as usual.

Handling the second case (when $\SSS_1$ is empty) requires to restore the top segment of $\SSS$ in $\SSS_1$. It turns out that the restoration process is the main bottleneck of this $\DFS$ algorithm. To aid the restoration process, Elmasry et al.\cite{ElmasryHK15} propose an elegant solution by maintaining an additional stack $\TT$, called a trailer stack. The top-most element of each segment in $\SSS$ is called as a \emph{trailer} element. The stack $\TT$ stores the trailer element of each segment in $\SSS$ -- except trailers of those segments which are already present in $\SSS_1$.

The stack $\TT$ is crucially used in the restoration process. Let $(u,u.\NXT)$ be the second top most entry in stack $\TT$. This implies that the first vertex of top segment of $\SSS$ is $\NN(u)[u.\NXT-1]$. Now, a $\DFS$-like algorithm is run starting from the vertex $\NN(u)[u.\NXT-1]$ to restore the top segment of $\SSS$ in $\SSS_1$ as follows: 

{\em Temporarily the meaning of $\GR$ and $\WH$ vertex is changed. Then, process $v \leftarrow \NN(u)[u.\NXT-1]$ to find $(v.\NXT-1)$ as follows: find the first $\GR$ neighbor $w$ of $v$, mark it $\WH$, push $(v,\ell+1)$ (where $\NN(v)[\ell]=w$),
and then start processing of $w$. Elmasry et al. \cite{ElmasryHK15} show that this restoration process correctly restores the top segment of $\SSS$.}

Some explanation is in order about the above procedure. Once we have found $v$, we want to find $v.\NXT$. Analogously, we can say that we want to find $v.\NXT-1$. This vertex, $w \leftarrow \NN(v)[v.\NXT-1]$, was a  $\WH$ vertex encountered while processing $v$. Due to $w$, we stopped the processing of $v$, put $(v,v.\NXT)$ on $\SSS$ and start the processing of $w$.

Even though the above algorithm is correct, it is still slow. Finding the first $\GR$ neighbor of a vertex $v$ takes  $O(\DD (v))$ time. To overcome this difficulty, Elmasry et al.\cite{ElmasryHK15} suggest the use of two more data-structures. The first data-structure $D$ is an array of size $n$ that contains the following information for each vertex $v$: if $v$ is an element of $\SSS$, then $D(v)$ contains
\begin{itemize}
  \item The segment number in which $v$ lies.
  \item The approximate position of  $v.\NXT-1$ in $\NN(v)$.
\end{itemize}

Since there are $\log n$ segments of $\SSS$ (as each segment is of size $O(n/\log n))$,  it requires $\log \log n$ bits to represent the first quantity. Similarly, storing the approximate position also takes $O(\log \log n)$ bits. Thus the space required for $D$ is $O(n \log \log n)$ bits.

The second term in $D(v)$ helps to fasten the search process for $v.\NXT-1$ only if the degree of $v$ is sufficiently small. However, to take care of  high degree vertices, the trailer stack $\TT$ is extended to include not only  trailers but also all  the pair $(u,u.\NXT)$, where $u$ is a high degree vertex. Finally, Elmasry et al. \cite{ElmasryHK15} show that the extended trailer stack $\TT$ takes $O(n)$ bits. Moreover, using $D$ and the extended $\TT$ restores  $\SSS_1$ correctly and efficiently.
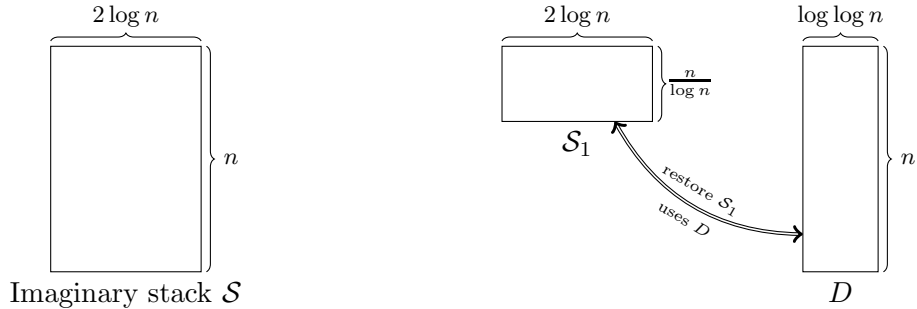
\begin{figure}
\begin{tikzpicture}
\draw [draw=black] (0,0) rectangle (2,3);
\node[below] at (1,0) {Imaginary stack $\SSS$};

\draw [decorate,decoration={brace,amplitude=3pt,mirror,raise=2pt},yshift=0pt]
(2,3) -- (0,3) node [black,midway,yshift=0.4cm] {\footnotesize
$2\log n$};

\draw [decorate,decoration={brace,amplitude=3pt,mirror,raise=2pt},yshift=0pt]
(2,0) --(2,3)   node [black,midway,xshift=0.4cm] {\footnotesize
$ n$};
\draw [draw=black] (6,2) rectangle (8,3);
\node[below] at (7,2) { $\SSS_1$};
\draw [decorate,decoration={brace,amplitude=3pt,mirror,raise=2pt},yshift=0pt]
(8,3) -- (6,3)  node [black,midway,yshift=0.4cm] {\footnotesize
$2\log n$};

\draw [decorate,decoration={brace,amplitude=3pt,mirror,raise=2pt},yshift=0pt]
(8,2) --(8,3)   node [black,midway,xshift=0.5cm] {\footnotesize
$ \frac{n}{\log n}$};

\draw [draw=black] (10,0) rectangle (11,3);
\node[below] at (10.5,0) { $D$};
\draw [decorate,decoration={brace,amplitude=3pt,mirror,raise=2pt},yshift=0pt]
(11,3) -- (10,3)  node [black,midway,yshift=0.4cm] {\footnotesize
$\log \log n$};

\draw [decorate,decoration={brace,amplitude=3pt,mirror,raise=2pt},yshift=0pt]
(11,0) --(11,3)   node [black,midway,xshift=0.4cm] {\footnotesize
$n$};

\draw [<->,double](7.5,2) to[bend right=30] node[above,rotate=-30] {\tiny{restore $\SSS_1$}} node[below,rotate=-30]
{\tiny{uses $D$}}
(10,0.5);
\end{tikzpicture}
\caption{A pictorial description of the approach in \cite{ElmasryHK15}. The restoration of stack $\SSS_1$ depends on $D$. The size of $D$ is $O(n \log \log n)$ and our aim is to reduce this size. Note that all the data-structures \cite{ElmasryHK15} are not shown in the figure.}
\end{figure}

\subsection{Our Approach}
We give a brief overview of our approach. In \cite{ElmasryHK15}, the array $D$ plays a critical role in the restoration process. While restoring the top segment, $D(v)$ provides the required information for each vertex $v$ which is a part of the top-most segment.  However, $D(v)$ takes $O(n \log \log n)$ bits -- a space we cannot afford. Our main observation is that we do not require information related to all vertices while restoring $\SSS_1$. Indeed, storing information about  vertices in the top-most segment  suffices. Unfortunately, it is not easy to keep information related to vertices in top-most segment efficiently in $O(n)$ space. To overcome this difficulty, along with the stack $\SSS_1$\footnote{In our algorithm, size of $\SSS_1$ is bit different than that in \cite{ElmasryHK15}. It is mentioned in Remark \ref{rem:remark1}} we implement $\SSS_2$ (a dynamic dictionary -- as described in Lemma \ref{lem:dict}) which contains  information about top vertices  $\frac{2n}{(\log \log n)^2}$ vertices of the imaginary stack $\SSS$. For each  vertex in $\SSS_2$,
we store $O(\log \log n)$ bits of information that will help us when we restore $\SSS_1$ (remember that the size of $\SSS_1$ is much less that the size of $\SSS_2$). We can show that the   size of $\SSS_2$ is $\approx O\Big(\frac{2n}{(\log \log n)^2} \times \log \log n\Big) = O\Big(\frac{n}{\log \log n}\Big)$ bits. Thus, we have successfully reduced the size of $\SSS_2$ (named $D$ in \cite{ElmasryHK15}).

Since $\SSS_2$ does not store the information of all the vertices in  stack $\SSS$, it faces the restoration problem as well. If   top $\frac{2n}{(\log \log n)^2}$ vertices are popped out of $\SSS$, those are also deleted from $\SSS_2$. Thus, we need to {\em restore} $\SSS_2$. To aid in the restoration of $\SSS_2$, we implement another data-structure $\SSS_3$, which contains the information  top $\frac{2n}{(\log \log \log n)^2}$ vertices of $\SSS$. For each vertex in $\SSS_3$, we will store $O(\log \log \log n)$ bits of information. The size of $\SSS_3$ can be shown to be $O(\frac{n}{\log \log \log n})$ bits. It is not hard to see that this process goes on recursively and we have many data-structures $\SSS_i$ where the last data-structure is $\SSS_{\log^*n}$.  $\SSS_{\log^*n}$ stores information about   $\frac{2 n}{\alpha^{2}}$ vertices, where $\alpha \ge 1$ is some constant. But, the restoration problem does not disappear yet. Now the question is how do we restore $\SSS_{\log^*n}$? Beyond this, we do not create any more data-structure. We restore $\SSS_{\log^*n}$ using the most trivial strategy, that is by running $\DFS$ all over again. Our main claim is that throughout our algorithm  $\SSS_{\log^* n}$ is restored at most $\alpha^2 $ times. We will show that the time taken to restore $\SSS_{\log^* n}$ is $O(m+n)$. Thus the total time taken to restore $\SSS_{\log^* n}$ is $O(\alpha^2(m+n)) = O(m+n)$ (since $\alpha$ is a constant). For other $\SSS_i$'s ($i \neq \log^* n$), our analysis is slightly different and it is the main technical contribution of this paper. We will show that the total time taken to restore  $\SSS_i$ over the entire course of the algorithm 
is $O(\frac{m}{\LL{i}}+n)$ where $\LL{i}  := \underbrace{\log \log \log \dots \log}_{i
\  \text{times}} n$. Thus, the time taken to restore all
$\SSS_i$'s over the entire course of the algorithm is $O(m + n \log^* n)$.

Let us now briefly describe the space taken by our algorithm.
Each $\SSS_i$ stores information about at most top $\frac{2n}{(\LL{i})^2}$ vertices of $\SSS$. Also, for each
such vertex, we will only store $O(\LL{i})$ bits.
Using succinct dictionary \cite{DemaineAPP06}, we will show that we can implement $\SSS_i$ in $O(\frac{n}{\LL{i}})$
space. Thus, the total space taken by our algorithm is $O\Big(\sum_{i=1}^{\log^* n} \frac{n}{\LL{i}}\Big) = O(n)$ bits.
Note that our algorithm will also use some other data-structures which we have not described till now. However, the
major challenge in our work was to bound the size of $\SSS_i$'s. All our other data-structures
 take $O(n)$ bits cumulatively. Thus, the total space taken by our algorithm is $O(n)$ bits.
This completes the overview of our algorithm. \\

\begin{figure}
\begin{tikzpicture}

\draw [draw=black] (6,2) rectangle (8,3);
\node[below] at (7,2) { $\SSS_1$};
\draw [decorate,decoration={brace,amplitude=3pt,mirror,raise=2pt},yshift=0pt]
(8,3) -- (6,3)  node [black,midway,yshift=0.4cm] {\footnotesize
$2\log n$};

\draw [decorate,decoration={brace,amplitude=3pt,mirror,raise=2pt},yshift=0pt]
(8,2) --(8,3)   node [black,midway,xshift=0.7cm] {\footnotesize
$ \frac{2n}{(\log n)^2}$};

\draw [draw=black] (10,1.5) rectangle (11,3);
\node[below] at (10.5,1.5) { $\SSS_2$};
\draw [decorate,decoration={brace,amplitude=3pt,mirror,raise=2pt},yshift=0pt]
(11,3) -- (10,3)  node [black,midway,yshift=0.4cm] {\footnotesize
$\log \log n$};

\draw [decorate,decoration={brace,amplitude=3pt,mirror,raise=2pt},yshift=0pt]
(11,1.5) --(11,3)   node [black,midway,xshift=0.8cm] {\footnotesize
$\frac{2n}{(\log \log n)^2}$};

\draw [<->,double](7.5,2) to[bend right=30] node[above,rotate=-4]
{\tiny{restore $\SSS_1$}} node[below,rotate=-4]
{\tiny{uses $\SSS_2$}}
(10,1.8);

\draw [draw=black] (13,1) rectangle (14,3);
\node[below] at (13.5,1) { $\SSS_3$};
\draw [decorate,decoration={brace,amplitude=3pt,mirror,raise=2pt},yshift=0pt]
(14,3) -- (13,3)  node [black,midway,yshift=0.4cm] {\footnotesize
$\log \log \log n$};

\draw [decorate,decoration={brace,amplitude=3pt,mirror,raise=2pt},yshift=0pt]
(14,1) --(14,3)    node [black,midway,xshift=1cm] {\footnotesize
$\frac{2n}{(\log \log \log  n)^2}$};

\draw [<->,double](10.8,1.5) to[bend right=30] node[above,rotate=-4]
{\tiny{restore $\SSS_2$}} node[below,rotate=-4]
{\tiny{uses $\SSS_3$}}
(13,1.2);

\draw [<-,double](13.8,1) to[bend right=30]
(15,0.5);

\draw [loosely dotted] (16,2) -- (18,2);

\draw [draw=black] (19,0) rectangle (20,3);
\node[below] at (19.5,0) { $\SSS_{\log^* n}$};
\draw [decorate,decoration={brace,amplitude=3pt,mirror,raise=2pt},yshift=0pt]
(20,3) -- (19,3)  node [black,midway,yshift=0.4cm] {\footnotesize
$\alpha$};

\draw [decorate,decoration={brace,amplitude=3pt,mirror,raise=2pt},yshift=0pt]
(20,0) --(20,3)    node [black,midway,xshift=0.4cm] {\footnotesize
$\frac{2n}{\alpha^2}$};

\draw [->,double](17.8,0.5) to[bend right=30]
(19,0.2);

\end{tikzpicture}
\caption{A pictorial description of our approach.
We implement many "stacks" $S_1, S_2, \dots, S_{\log^* n}$.
The restoration of $\SSS_i$ uses $\SSS_{i+1}$ as $\SSS_{i+1}$ contains the vertices to be restored in $\SSS_i$.   Note that all our data-structures 
are not shown in the figure.}
\end{figure}
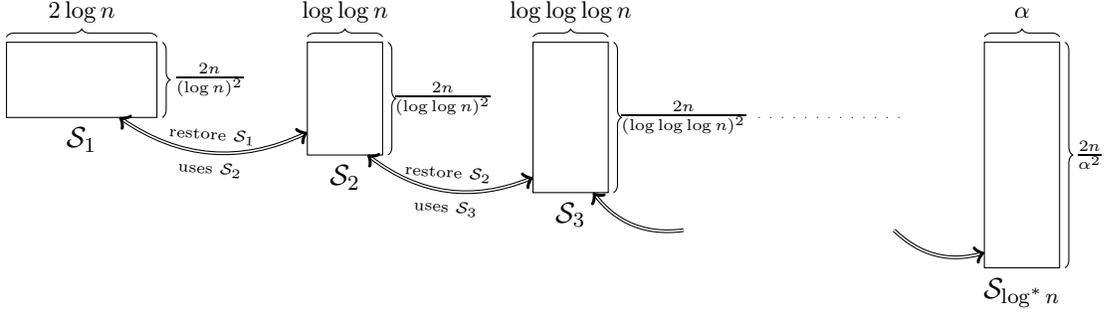
\begin{remark}
	\label{rem:remark1}
In the above description, each $\SSS_i$ contains at most
top $\frac{2n}{(\LL{i})^2}$ elements of $\SSS$. Thus, the size of $\SSS_1$ is $\frac{2n}{(\log n)^2}$. This is a crucial difference from the Elmasry et al. \cite{ElmasryHK15} algorithm, where the size of $\SSS_1$ was $\frac{2n}{\log n}$. The main reason for this change is to decreases the space taken by our algorithm. Indeed, the cumulative space taken by all $\SSS_i$'s (in our algorithm)  can be shown to be  $\sum \frac{n}{\LL{i}} =O(n)$. In spite of this change, the running time of our algorithm does not suffer. To summarize, this is an important technical change from the previous work with the sole aim to decrease the space taken by the algorithm.
\end{remark}

 \section{Preliminaries}
\label{sec:prelim}
  In our algorithm, the following  data-structure plays a crucial role.
  \begin{lemma} (Succinct Dynamic Dictionary \cite{DemaineAPP06})
  \label{lem:dict}
  Given a universe ~$\UU$ of size $u$, there exists a dynamic
  dictionary that stores a subset $\SSS \in \UU$ of size
  at most $n$. Each element of $\UU$ has a satellite data
  of size $r$ where $r \in O(\log n)$.
The time taken for membership, retrieval, insert, and delete any element
   (and its satellite data) is  $O(1)$ with probability $(1-1/n^c)$ for some chosen constant $c$.
  The space taken by the data-structure is $ n \log \frac{u}{n}
  + nr$ bits.
  \end{lemma}

Note that a similar dictionary was also described  in
Lemma 2.1 of \cite{ElmasryHK15}.

We define few basic notation/data-structures that will be used in the ensuing discussion.

\begin{itemize}
\item $\LL{i}  := \underbrace{\log \log \log \dots \log}_{i \  \text{times}} n$.

\item $\log^* n$ (iterated logarithm) is the
number of times the logarithm function is iteratively applied till the result is $\le 2$.
 Define $\alpha := \underbrace{\log \log \dots  \log}_{\log^* n \ \text{times}} n = \LL{\log^* n}$. Note that $1 < \alpha  \le 2.$
\item We divide the imaginary stack $\SSS$ into segments of size $\Big\lceil\frac{n}{(\log n)^2}\Big\rceil$. An $\IS{i}$ ($1 \le i \le \log^* n$) contains vertices
of  $\Big(\frac{\log n}{\LL{i}}\Big)^2$ consecutive segments
of $\SSS$. We divide the imaginary stack $\SSS$ into $\IS{i}$s
from
bottom to top (only the topmost $\IS{i}$ may contain less
number of consecutive segments).  The total number of vertices
in an $\IS{i}$ is at most
$\Big\lceil \Big(\frac{\log n}{\LL{i}^{}}\Big)^2\times \frac{n}{(\log n)^2}\Big\rceil = \Big\lceil\frac{n}{(\LL{i})^{2}}\Big\rceil$
and the total number of $\IS{i}$s is at most $(\LL{i})^{2}$.
For brevity, we will drop the ceil notation in the rest of the paper.

\item Stack $\SSS_1$

  A stack $\SSS_1$ will store the vertices present in at most top
  two segments of $\SSS$. Each cell of $\SSS_1$ contains
the
  tuple of type $(v,v.\NXT)$.

  \item Dynamic  Dictionary for $\SSS_i$ $(2 \le i \le \log^* n)$

 We will  store information about vertices of at most top  two $\IS{i}$ in a dynamic dictionary   $\SSS_i$ $(2 \le i \le \log^*n )$.\ This information will be crucial in restoring $\SSS_{i-1}$.

 \item Trailers

  In \cite{ElmasryHK15}, the restoration algorithm uses the trailer stack to find a  vertex from which the restoration of $\SSS_1$ should start. In our algorithm, as we have to restore $\SSS_1, \SSS_2, \dots , \SSS_{\log^* n}$, we require many trailer stacks.

 To this end, we implement a trailer stack for each $\SSS_i$.
  In the trailer stack $\TT_i$ ($1 \le i \le (*))$, we keep
  the bottommost element of the imaginary stack $\SSS$ and the
  top vertex of all  $\IS{i}$s of  $\SSS $  that are not present in $\SSS_i$.

\end{itemize}

\section{Our Algorithm}
\label{sec:algo}

Our algorithm is nearly similar to the  Elmasry et al.\cite{ElmasryHK15}
algorithm.  We initially color all the vertices white
(the space taken by the $\CC$ array is $O(n) $ bits as we
color a vertex $\WH$, $\GR$ or $\BL$ only).
Then we take an arbitrary vertex, say $u$, and do a $\DFS$
from $u$. Like Elmasry et al.\cite{ElmasryHK15}, initially
$(u,1)$ is pushed on to the stack. Additionally, we  also
 insert $(u,1)$ to all other $\SSS_i$'s.

\begin{algorithm}[h]
  \caption{ $\DFS(u)$ }
  \label{alg:maindfs}
  \For{$i \leftarrow \log^* n$ to 1}
  {
      $\IN(u,i,1)$\;
  }

  \While{trailer $\TT_1$ is not empty }
  {

     $(v,v.\NXT) \leftarrow $ top element of $\SSS_1$\;
     \For{ $i \leftarrow \log^* n$ to 1}
     {
         $\DL(v,i)$\;
     }

     $\CC(v) \leftarrow \GR$\;
     \While{$v.\NXT \le \DD(v)$}
     {

         \If{ $\CC(\NN(v)[v.\NXT])$ is \WH }
         {
             \For{$i \leftarrow \log^* n$ to 1}
             {

                 $\IN(v,i,v.\NXT+1)$\;
                 $\IN(\NN(v)[v.\NXT],i,1)$\;
             }
             break\;
         }
         \Else
         {
             $v.\NXT \leftarrow v.\NXT+1$\;
         }

     }
     \Else
     {
         $\CC(v) \leftarrow \BL$ \;
     }

  }
\end{algorithm}

We then go over the stack $\SSS_1$ till it becomes empty.
Analogously, we can say that we will process the stack $\SSS_1$
till the trailer $\TT_1$ becomes empty --- as $\TT_1$
always contains the bottommost element of the imaginary stack
$\SSS$. Our $\textsc{While}$ loop is similar to the standard
$\DFS$ algorithm with the addition that we push and pop
not only to $\SSS_1$ but insert to and delete from all $\SSS_i$'s.
Let  $(v,v.\NXT)$ be the top element of $\SSS_1$. We pop $v$ from $\SSS_1$ and also delete
it from all other $\SSS_i$'s.
Then we color $v$ gray. 
 We then check if the $(v.\NXT)$-th
neighbor of $v$, $\NN(v)[v.\NXT],$ is white or not. If it
is white, then we first push $v$ back on to the stack (and all other $\SSS_i$'s). After that, $\NN(v)[v.\NXT]$ is pushed to $\SSS_1$ and all the other
relevant data-structure.
When we have processed all the neighbors of $v$, it is colored
$\BL$.

We now calculate the running time of our $\DFS$ algorithm in Algorithm \ref{alg:maindfs}.
In the classical $\DFS$ algorithm,  a $\GR$
vertex is pushed  onto the stack again after it finds a new $\WH$ vertex. This implies that vertices can be pushed on to the stack at most $O(n)$ times.
Our $\DFS$ algorithm is nearly similar to the classical
$\DFS$ algorithm with the only difference that we insert/delete
into $\log^* n$ ``stacks" instead of one.   Thus we claim
the following running time:

\begin{lemma}
\label{lem:maindfs}
Not accounting for the time taken by $\IN$ and $\DL$ procedures,
the  time taken by our $\DFS$ algorithm in Algorithm \ref{alg:maindfs} is $O(m+n\log^* n)$.
\end{lemma}

In $\IN(v,i,v.\NXT)$  procedure, we add the information about vertex $v$ to $\SSS_i$. Remember that  $\SSS_i$ is used to restore $\SSS_{i-1}$. We will now describe $\SSS_i$ in detail.

\section{ Information in $\SSS_i$}
\label{sec:info}
In \cite{ElmasryHK15}, where we just have to restore $\SSS_1$, the following two pieces of information about each vertex is stored in $D$: (1) The segment number in which $v$ lies. (2) The approximate position in $\NN(v)$ where $v.\NXT-1$
lies.

We try to generalize this idea. Unlike $D$, the dictionary $\SSS_i$ in our algorithm contains information about vertices present in at most two top $\IS{i}$s. For each such $v \in \SSS_i$, let $\SSS_i(v)$ denote the cell in which information related to $v$ is stored. We will store the following information related to $v$.
\begin{enumerate}
  \item The $\IS{(i-1)}$ number in which $v$ lies.

Remember that $\SSS_i$'s main function is to
restore $\SSS_{i-1}$.  Thus, for each vertex $v$, we will
store the $\IS{(i-1)}$ to which $v$ belongs, let us denote it by $\SEG_{i-1}(v)$.
 $\SEG_{i-1}(v)$ will help the restore algorithm of $\SSS_{i-1}$ to
check whether $v$ indeed lies in the top $\IS{(i-1)}$.
Since
 the total number of $\IS{(i-1)}$ is
$(\LL{i-1})^{2}$, $2\LL{i}$ bits are required to represent $\SEG_{i-1}(v)$.

  \item  The approximate position in $\NN(v)$ where $v.\NXT-1$
lies.

The above information is used to find $v.\NXT-1$ efficiently. It would have been nice if we could explicitly store $v.\NXT-1$. However, this will require $O(\log n)$ bits for each vertex in $\SSS_i$ --- a space which we cannot afford. To overcome the space limitation, we divide $\NN(v)$ into groups of appropriate size and store the group number in which $v.\NXT-1$ lies.
\end{enumerate}

The exact definition of the second term requires some more work.  Note that $\SEG_{i-1}(v)$ takes just $O(\LL{i})$ bits. We want the second term also to take $O(\LL{i})$ bits. Thus, the number of groups into which we divide $\NN(v)$ should not be huge (it should be $\le \LL{i-1}$).  However, if the number of groups is small, it implies that the {\em group size}, i.e., the number of vertices in each group, may be large. Thus, given the group number,  finding $v.\NXT-1$ in the group will take more time.  Thus, we are faced with a dilemma where reducing the space increases the running time of our algorithm. To overcome this dilemma, we extend a  strategy used in \cite{ElmasryHK15}.  Elmasry et al. \cite{ElmasryHK15} divided the vertices into two sets -- {\em heavy} and {\it light}. A light vertex has {\it low} degree --- thus, its group size is small. For heavy vertices, they show that the total number of heavy vertices is small and for each heavy vertex $v$, $v.\NXT-1$ can be stored explicitly without using too much space. We plan to extend this strategy. But unlike \cite{ElmasryHK15}, we have a hierarchy of heavy and light vertices (since we have a hierarchy of $\SSS_i$'s).

\subsection{Light Vertices}

\begin{definition}
A vertex $v$ is $\IL{i}$ if $\DD(v) \le \frac{m(\LL{i-1})^2}{n}$
where $2 \le i \le \log^* n$. We define all the vertices in
$V$ to be $\IL{1}$.
\end{definition}

 We are now ready to define the second information related to $v$  stored in $\SSS_i$. If $v$ is $\IL{i}$, then we divide $\NN(v)$ into groups of size
$\frac{\DD(v)}{(\LL{i-1})^{3}}$.

\begin{definition}
If $v$ is $\IL{i}$, then the second information of $v$ (approximate position of $v.\NXT-1$ in $\NN(v)$) stored in $\SSS_i$ is $\GRO_{i-1}(v.\NXT-1)$ defined as follows:  $\GRO_{i-1}(v.\NXT-1) :=\EL$ if $\EL \frac{deg(v)}{(\LL{i-1})^{3}}
< v.\NXT-1 \le (\EL+1) \frac{deg(v)}{(\LL{i-1})^{3}}$.
\end{definition}

The total number of groups
of $\NN(v)$ is  $(\LL{i-1})^{3}$. Thus the total number of bits
required to represent $\GRO(v.\NXT-1)$ is 3$\LL{i}$ bits.

Remember that we partitioned the set of vertices into light and heavy only to make the group size small. We now bound the number of vertices in a  group  of a $\IL{i}$ vertex.\begin{observation}
\label{obs:lightgroupsize}
If $v$ is $\IL{i}$, then the total number of vertices in each group of $\NN(v)$ is
$\le \frac{\DD(v)}{(\LL{i-1})^{3}} \le \frac{m (\LL{i-1})^{2}}{n(\LL{i-1})^{3}} = \frac{m}{n\LL{i-1}}$.
\end{observation}

We are now ready to formally define the information about vertex $v$ stored in $\SSS_i$. \begin{itemize}
\item If an $\IL{i}$ vertex  $v$ becomes a part of top 
$\IS{i}$ of imaginary stack $\SSS$, then we store the following
information about $v$.  $\SSS_i(v) = (\SEG_{i-1}(v),\GRO_{i-1}(v.\NXT-1))$

\item If vertex $v$ is not $\IL{i}$, then $\SSS_i(v)
= (\SEG_{i-1}(v),0)$, that is we just store the $\IS{ (i-1)}$
in which $v$ resides.

\end{itemize}
Some explanation is in order. If $v$ is an $\IL{i}$ vertex, then we can store the information $(\SEG_{i-1}(v),\GRO_{i-1}(v.\NXT-1))$
corresponding to $v$. We have already shown that both these terms take $O(\LL{i})$ bits. Moreover, given the group number $\GRO_{i-1}(v.\NXT-1)$, we can find $v.\NXT-1$ in $O\Big(\frac{m}{n\LL{i-1}}\Big)$ time,
as the number of vertices in each group of an $\IL{i}$ vertex is $\le \frac{m}{n\LL{i-1}}$ (using Observation \ref{obs:lightgroupsize}).

However, if $v$ is not $\IL{i}$, then its group size may be
$> \frac{m}{n\LL{i-1}}$  which is not desirable (as this might increase the search time for $v.\NXT-1$).  So, for such a
vertex,  we store  $\SEG_{i-1}(v)$ only as there is no point
in storing the second term (the second term 0 is just a
dummy term). But for efficiency, we need to store some information regarding $v.\NXT-1$ even for the vertex which is not $\IL{i}$.
In the next section, we describe a data-structure which will efficiently store information about all  non $\IL{i}$ vertices.

\subsection{Heavy Vertices}

\begin{definition}
A vertex $v$ is $\IH{i}$ if  $\frac{m (\LL{i-1})^2}{n}<
\DD(v) \le \frac{m (\LL{i-2})^2}{n}$ where $3 \le i \le
 \log^* n$.
We define a $\IH{2}$ vertex separately. A vertex $v$ is
said to be $\IH{2}$ if $\frac{m(\log n)^{2}}{n} < \DD(v)
\le n$.
\end{definition}

Note that our definition partitions the vertex set nicely.
We prove this nice property in the following lemma:

\begin{lemma}
\label{lem:caseheavy}
If $v$ is not $\IL{i}$ $(i\ge 2)$, then it is $\IH{j}$
for some $j$ where
$2 \le j \le i$.
\end{lemma}
\begin{proof}
Since $v$ is not $\IL{i}$, $\frac{m(\LL{i-1})^{2}}{n}
< \DD(v)
\le n $.
Thus, there exists a $j$ ($3 \le j \le i$) such that $\frac{m(\LL{j-1})^{2}}{n}
\le \DD(v) < \frac{m (\LL{j-2})^{2}}{n}$ or $\frac{m(\log
n)^{2}}{n}
\le \DD(v) < n$ (the case when $j=2$).
\end{proof}

We store the information related to an $\IH{i}$ vertex in a dynamic dictionary $\HH_i$ where $i \ge 2$.  Since degree of a $\IH{i}$ vertex
$v$ is $\ge   \frac{m (\LL{i-1})^{2}}{n}$, total number of $\IH{i}$
vertices is $O\Big(\frac{n}{(\LL{i-1})^{2}}\Big)$. Similar to $\IL{i}$
vertices, we divide $\NN(v)$ into groups of size $\frac{\DD(v)}{(\LL{i-2})^{3}}$.
The only problem with this group size is that it is not
defined for $i=2$. If $i=2$, then we divide $\NN(v)$ into groups of size 1.

 We store the group number of $v$ in the   dynamic dictionary
  $\HH_i$, that is $\GRO_{i-2}(v.\NXT-1)$ defined as follows:  $\GRO_{i-2}(v.\NXT-1)
:=\EL$ if $\EL \frac{deg(v)}{(\LL{i-2})^{3}}
< v.\NXT-1 \le (\EL+1) \frac{deg(v)}{(\LL{i-2})^{3}}$.
 Since we divide
  $\DD(v)$ into groups of size $\frac{\DD(v)}{(\LL{i-2})^{3}}$,
  the total number of groups is $(\LL{i-2})^{3}$.  This implies
  that total space required to represent the group number
  per vertex in $\HH_i$ is $3\LL{i-1}$ bits.

Using Observation \ref{obs:lightgroupsize}, if a vertex $v$ is $\IL{i}$, then  the associated
group size\ (stored in $\SSS_i$) is $ \frac{m}{n\LL{i-1}}$. The next lemma
present a very crucial feature of our algorithm:

\begin{lemma}
\label{lem:groupsize}
Let $v$ be a vertex in $\SSS_i$, then the group size associated
with $v$ is of size $\le 1+ \frac{m}{n\LL{i-1}}$.
\end{lemma}
\begin{proof}
If $v$ is $\IL{i}$, then we have already seen that the group
size associated with $v$ (and stored in $\SSS_i$)  is $ \frac{m}{n\LL{i-1}}$.
Using Lemma \ref{lem:caseheavy}, if $v$ is not $\IL{i}$,
then it is $\IH{j}$ for $2\le j \le i$. Thus, the information
about the group of $v$ is stored in $\HH_j$, that is $\GRO_{j-2}(v.\NXT-1)$.
To this end, we divide $\NN(v)$ into group of size $\frac{\DD(v)}{(\LL{j-2})^{3}}$.
There are two cases:

\begin{enumerate}
\item $j > 2$

Since $v$ is $\IH{j}$, $\DD(v) \le \frac{m( \LL{j-2})^{2}}{n}$.
This implies that the size of each group is $\le\frac{m}{n\LL{j-2}} \le  \frac{m}{n\LL{i-1}}$.

\item $j=2$

 By definition, the group size is exactly $1$. \end{enumerate}

Thus, the group size associated with $v$ is  $\le 1+ \frac{m}{n\LL{i-1}}$.

\end{proof}

The above lemma shows a crucial property of all  vertices
in $\SSS_i$. The associated group size of all these vertices is $\le 1+ \frac{m}{n\LL{i-1}}$ irrespective
of their degree. Thus, whenever we are searching for $v.\NXT-1$ for a vertex $v$, we have to search atmost $1+ \frac{m}{n\LL{i-1}}$. We will crucially exploit this property in the restoration algorithm. However, before that let us take a look at the insert and delete procedures.

\section{Insert and Delete Procedures}

\begin{algorithm}[h]
  \caption{ $\IN(v,i,v.\NXT)$ }
  \label{}
   \If{  $|\SSS_i| = \frac{2n}{(\LL{i-1})^{2}}$    }
      {
           $\REF(i)$\;
      }
  \If{$v$ is $\IL{i}$}
  {

      $\SSS_i.\IN(v,(\SEG_{i-1}(v),\GRO_{i-1}(v.\NXT-1)))$
or  $\SSS_1.\PU(v,v.\NXT)$ (if $i=1$)\;
  }
  \Else{
    $\SSS_i.\IN(v,(\SEG_{i-1}(v),0))$
  }
        \If{ $\TT_i$ is empty  or recently pushed element
becomes the top element of an $\IS{i}$  }
      {
          $\TT_i.\PU(v,v.\NXT)$\;
      }
  \If{$v$ is $\IH{i}$}
  {
      $\HH_i.\IN(v,\GRO_{i-2}(v.\NXT-1))$\;
  }
\end{algorithm}
   In the $\IN$ procedure,
$v$ is to be inserted
in $\SSS_i$. But $\SSS_i$ may be full, that is, it has $\frac{2n}{(\LL{i})^2}$
vertices. So, we call $\REF(i)$ which basically aims at
removing
half of the elements of $\SSS_i$.  After
the restoration, $\SSS_i$ has the top $\frac{n}{(\LL{i})^{2}}$
 vertices of the imaginary stack $\SSS$. We then insert
$(\SEG_{i-1}(v),\GRO_{i-1}(v.\NXT-1))$
in $\SSS_i$. If this newly added element becomes the top
element of an $\IS{i}$ or the trailer itself is empty then
we add $(v,v.\NXT)$  to the trailer
$\TT_i$. Lastly, if $v$ is $\IH{i}$, then it is added to
$\HH_i$. Three details are missing from the  pseudo
code of $\IN$. We list them now:

\begin{enumerate}
\item {\em Calculating $\SEG_{i-1}(v)$}

Let $k_1$ be the total number of vertices in trailer $\TT_i$
and $k_2$ be the total number of vertices in $\SSS_i$.
We first calculate the total number of vertices below $v$
in the imaginary stack $\SSS$. This is $k= (k_1-1)\times
\text{size of $\IS{i}$ }+k_2$ = $(k_1-1) \times \frac{n}{(\LL{i})^{2}}
+k_2$. Once we have calculated $k$, finding $\SEG_{i-1}(v)$
is just a mathematical calculation.

\item {\em Calculating  $\GRO_{i-1}(v.\NXT-1)$ or  $\GRO_{i-2}(v.\NXT-1)$
}

This is just a mathematical calculation once we know $v.\NXT$
and $\DD(v)$.

\item {\em Finding if $v$ is a top element of an $\IS{i}$}

This can be done by maintaining the number of elements currently present
in the imaginary stack $\SSS$. Before inserting $v$, if
$|\SSS| = 0$ or $|\SSS| = \frac{cn}{(\LL{i})^{2}}-1$
$(c \ge 1)$, then we insert $(v,v.\NXT)$ on to the trailer
$\TT_i$.

\end{enumerate}

\begin{algorithm}[h]
  \caption{ $\DL(v,i)$ }
  \label{}
      \If{$|\SSS_i| < \frac{n}{2(\LL{i-1})^{2}}$  and $\TT_i$
has
at least two elements }
      {
          $\REE(i)$\;
      }
      \If{$v$ is $\IH{i}$}
      {
          $\HH_i.\DL(v)$\;
      }
  \If{$v$ is on the top of the trailer $\TT_i$}
  {
     $\TT_i.\PO()$\;
  }
      return $\SSS_{i}.\DL(v)$ or $\SSS_1.\PO()$ (if $i=1)$

\end{algorithm}

The $\DL(v,i)$ is nearly similar to the $\IN$ procedure.
We first check if the number of elements in $\SSS_i$ is
 less. If yes, then  we also have to check if
the trailer itself has enough elements. If yes,  then
we call $\REE(i)$. After its execution,  $\SSS_i$ contains
topmost  $\frac{n}{(\LL{i})^{2}}$ vertices of the imaginary
stack $\SSS$. If $v$ is $\IH{i}$, then it is removed from $\HH_i$.
 After this, the top element
of $\SSS_i$ (and $\TT_i$ if necessary) is removed.

The following lemma about the running time of $\IN$ and
$\DL$ is immediate (due to our data-structure in Lemma
\ref{lem:dict}).
\begin{lemma}
        \label{lem:ins_del_O1}
Apart from the time taken by $\REE$ and $\REF$, the running
time taken by $\IN$ and $\DL$ procedure is $O(1)$ with high
probability\footnote{Since we use the data-structure described in Lemma \ref{lem:dict}
at most poly($n$) times, all insert and deletes are successful
with probability $\ge 1-\frac{1}{n^{c}}$ where $c$ is some constant. }.
\end{lemma}

\section{Restore Procedure}
\begin{algorithm}[h]
  \caption{ $\REE(\log^* n)$ }
  \label{}

  color all the $\GR$ vertices $\WH$\;

  $u \leftarrow$ the vertex from which we started our $\DFS$
in $\SSS$.

$v \leftarrow u$\;

  \Do{ $v$ is not equal to the top of trailer $\TT_{\log^*
n}$}
  {
      
      Let $k$ be the index of the first white neighbor of
$v$\;

$\CC(v) \leftarrow \GR$\;
\If{ $v$ lies in the top $\IS{(\log^* n )}$ }
      {
          \If{$v$ is $\IL{\log^* n}$ }
          {
              $\SSS_{\log^* n}.\IN(v,(\SEG_{\log^* n-1}(v),\GRO_{\log^*
n-1}(k))$\;
          }
          \Else
          {
               $\SSS_{\log^* n}.\IN(v,(\SEG_{\log^* n-1}(v),0))$\;
          }

      }

      $v \leftarrow \NN(v)[k]$\;
  }

\end{algorithm}
We now move on to the most important part of our algorithm,
that is the restoration of $\SSS_i$'s. First, we describe our
approach for restoring the last dictionary, that
is, $\REE(\log^* n)$. Remember that to restore the last dictionary, we do the most trivial thing, that is run the $\DFS$ algorithm again.
So, we run the $\DFS$ algorithm again from the starting vertex $u$
ignoring all the black vertices (this process is similar to the one described in \cite{ElmasryHK15}). We mark all the $\GR$ vertices
$\WH$ and perform a $\DFS$ from $u$ till we hit the topmost trailer of 
$T_{\log^* n}$.  Whenever we encounter a vertex
of the top $\IS{\log^* n}$, we add it to $\SSS_{\log^* n}$ after calculating relevant parameters (as similar to that in $\IN$ algorithm).
Note that we can easily find if $v$ is a part of top $\IS{\log^* n}$
by comparing the number of vertices processed by the  restore
algorithm to the number of elements in  the imaginary stack
$\SSS$\ (which we can easily maintain).
We now show that our $\REE(\log^* n)$ procedure is correct. To
this end, we will compare our algorithm with the $\DFS$ algorithm that works with the imaginary stack $\SSS$. We will call this $\DFS$ algorithm as an {\em imaginary $\DFS$ algorithm}. We first observe the following:
\begin{observation}
\label{obs:fact}
Let $(v, v.\NXT)$ be an entry on the imaginary stack $\SSS$
when we call $\REE(\log^* n)$. Then, all  vertices in $\NN(v)[1
\dots v.\NXT-2]$ are black or gray when the imaginary $\DFS$
algorithm pushes this entry on to $\SSS$.
\end{observation}
\begin{proof}
Consider the step when the imaginary $\DFS$ algorithm pushes
the entry $(v,v.\NXT)$ on to the stack. This means that
it has found a white vertex $\NN{(v)[v.\NXT-1}]$. Thus,  $v$ has already processed all vertices in $\NN(v)[1 \dots v.\NXT-2]$
and color of each processed vertex is either gray or black.
\end{proof}

We now use the above observation to prove that $\REE(\log^* n)$
is correct.

\begin{lemma}
 Let $(v,v.\NXT)$ be an entry on the imaginary stack $\SSS$
when we  call $\REE(\log^* n)$. Then, (1) $\REE(\log^* n)$ also processes
the tuple $(v,v.\NXT)$ and (2) color of all the non-black
vertices is exactly same in the imaginary $\DFS$ algorithm
 and our $\REE$
algorithm (after  both algorithms process $v$).

\end{lemma}

\begin{proof}
First, note that we start our restoration process without
touching the color of a black vertex. Thus, if a vertex
is black
in the imaginary $\DFS$ algorithm (at the time we call $\REE(\log^* n)$),
it is also black in our
algorithm.

We now prove the statement of the lemma using induction.
Consider the moment when the imaginary $\DFS$ algorithm
put the entry $(u,u.\NXT)$ on to imaginary stack  $\SSS$
where $u$ is the vertex with which we started our $\DFS$.
 We now claim that there is no $\GR$ vertex in the graph
at this point in the imaginary $\DFS$ algorithm. This is
because all the gray vertices are always on
the imaginary stack $\SSS$ and when $u$ is processed, there
are no vertices on the imaginary stack. Thus, all  non-black vertices have $\WH$ color before the first push. Now, we claim that  (1) is true.
This is because the color of all the vertices
in $\NN(u)[1 \dots u.\NXT-2]$ is black, thus same for both  algorithms.
Due to Observation \ref{obs:fact}, we correctly find $(u,u.\NXT)$.
Before pushing$(u,u.\NXT)$ on to the stack, both
the algorithms make $u$ $\GR$. After the processing
of $u$, both the algorithms have same colors for all the non-black
vertices, thus (2) is also true.

We now show that the statement is true in general when we are inserting an element $(v, v.\NXT)$ at the $k^{\text{}th}$ iteration. Using the induction hypothesis, all the non-black vertices have same color at the end of the $(k-1)$-th iteration. Also, if a vertex is black in the imaginary $\DFS$ algorithm, it is also black at the start of our restore algorithm (since we donot touch black vertices). Since the imaginary $\DFS$ algorithm puts $(v, v.\NXT)$ on to the stack, vertices in  $\NN(v)[1 \dots v.\NXT-2]$ are black or gray. Using  the above arguments, the color of these vertices is same even in our algorithm. Thus, we also push $(v, v.\NXT)$ in our algorithm. Thus, (1) is true. Before pushing $(v,v.\NXT)$, both our algorithm and the imaginary $\DFS$ algorithm
mark $v$ $\GR$ -- the only change in the color of a vertex. Thus even (2) is true. This completes the induction step.
\end{proof}
The above lemma implies that at the end of the restoration,
$\SSS_{\log^* n}$ contains vertices from the top $\IS{(\log^* n)}$
of $\SSS$ and the color of each vertex is also correctly
restored. In the restoration process, we use the data-structure described
in Lemma \ref{lem:dict}
at most poly($n$) times, thus all insert and deletes are successful
with probability $\ge 1-\frac{1}{n^{c}}$ where $c$ is some
constant. Thus, the algorithm succeeds with very high probability.

Since, we are basically running the imaginary $\DFS$ again
to restore $\SSS_{\log^* n}$, the following lemma is immediate.
\begin{lemma}
\label{lem:restorelast}
The time taken to restore $\SSS_{\log^* n}$ is $O(m+n)$ with
high probability.

\end{lemma}

\begin{algorithm}[h]
  \caption{ $\REE(i)$ }
  \label{}

  $(w,w.\NXT) \leftarrow \text{second top element in}\ \TT_i$\;

  $v \leftarrow \NN(w)[w.\NXT-1]$\;

  \Do{ $v$ is not equal to the top of trailer $\TT_i$}
  {

          $(\SEG_{i}(v),l_{v}) \leftarrow \SSS_{i+1}.\GET(v)$\;
        $k \leftarrow l_{v} \frac{\DD(v)}{\LL{i}} $\;

      \If{  $v$ is \IH{j} where $j \le i$}
      {
           $ l_{v} \leftarrow \HH_{j}.\GET(v)$\;

           $ k \leftarrow l_{v} \frac{\DD(v)}{\LL{j-2}}$ or $l_v$ (if $j=2$)

      }

      \For{$k' = k$ to
$k+1+ \frac{m}{n \LL{i}}$  }
      {
          $x \leftarrow \NN(u)[k']$\;
          \If { $x$ is $\GR$ and $x$ in present in $\SSS_{i+1}$}
          {
              $(\SEG_i(x),l_{x})\leftarrow \SSS_{i+1}.\GET(x)$\;

              \If{$\SEG_i(v) = \SEG_i(x)$}
           {
               break\;
           }
          }

      }

      \If{$v$ is $\IL{i}$}
      {
          $\SSS_i.\IN(v,(\SEG_{i-1}(v),\GRO_{i-1}(k')))$\;
      }
      \Else
      {
          $\SSS_i.\IN(v,(\SEG_{i-1}(v),0))$\;
      }
      $v \leftarrow x$\;
      $\CC(v) \leftarrow \WH$\;
  }

  recolor all $\WH$ colored vertex during the above while
loop $\GR$ again\;
\end{algorithm}

Let us now look at $\REE(i)$ where $1 \le i < \log^* n$. Before $\REE(i)$
is called, we will assume that $\SSS_{i+1}$ has {\em enough} elements. This assumption is required as the vertices to be restored in $\SSS_i$ need to be present in $\SSS_{i+1}$.
\begin{itemize}
\item  $\SSS_{i+1}$ contains at least
$\frac{n}{2(\LL{i+1})^{2}}$ vertices (we will prove this crucial
assumption in the analysis)

\end{itemize}

For restoring $\SSS_i$,  we start from the second element from top in 
trailer $\TT_i$ and basically try to run the $\DFS$-like algorithm
from it. Let
$(w,w.\NXT)$ be second element from top in
trailer $\TT_i$. It implies that the first vertex (to be restored) in $\SSS_i$
is $\NN(w)[w.\NXT-1]$.
So, we start a $\DFS$ from $\NN(w)[w.\NXT-1]$ with one simple change (similar
to Elmasry et al.\cite{ElmasryHK15}) --
we change the meaning of $\WH$ and $\GR$
vertices. This is because all the vertices
to be restored in  $\SSS_i$ are $\GR$
and should not be processed once they are added in $\SSS_i$.

Let $v \leftarrow \NN(w)[w.\NXT-1]$.
Since the size of $\SSS_{i+1}$ is sufficiently larger than
$\SSS_i$, $v$ is present in $\SSS_{i+1}$.
Using $\SSS_{i+1}$, we find the $\IS{i}$ number to which
$v$ belongs. In addition, we also want
to find   $v.\NXT-1$. To this end, we check
if $v$ is $\IL{(i+1)}$. If yes, then we can find $l_v =
\GRO_i(v.\NXT-1)$, that is the approximate group in which
$v.\NXT-1$ resides. However,
if $v$ is not $\IL{(i+1)}$, then we use Lemma \ref{lem:caseheavy}
to conclude that $v$ is $\IH{j}$ for some $j \le i+1$, and
we find  $l_v = \GRO_{j-2}(v.\NXT-1)$
where $j \le i+1$. By Lemma \ref{lem:groupsize}, irrespective
of the fact whether $v$ is $\IL{(i+1)}$ or $\IH{j}$, the
group  in which $v.\NXT-1$ lies contains at most 1+$\frac{m}{n \LL{i-1}}$ vertices.
Now comes the most important part of our algorithm. We want
to identify $v.\NXT-1$ correctly once we have found the
group in which $v.\NXT-1$ resides. We will now use the following
lemma which will help us in identifying $v.\NXT-1$.

\begin{lemma}
\label{lem:nextcorrect}
Let $l_v$ be the group number that was found out in $\REE(i)$
procedure while processing $v$. Then $v.\NXT-1$ is the index of the first $\GR$
vertex, say $x$,  in this group such that  $\SEG_{i-1}(x)$ is equal to $\SEG_{i-1}(v)$.
\end{lemma}
\begin{proof}
We know that $v.\NXT-1$ lies in the group $l_v$. Let $x
\leftarrow \NN(v)[v.\NXT-1]$. We first discuss the properties
of vertex $x$. Since we are restoring the top $\IS{i}$,
$x$ should lie in the same segment as $v$, that is $\SEG_{i}(v)
= \SEG_{i}(x)$. In the imaginary $\DFS$ algorithm, consider
the step at which $v$ discovers $x$. Using Observation \ref{obs:fact},
we claim that at that point $x$ is the first $\WH$ vertex of the group.
Indeed, if there is another white vertex lying before $x$
in $\NN(v)$, then that vertex will be processed first by
the imaginary $\DFS$ algorithm.

Since the meaning  of $\WH$ and $\GR$ vertices are changed during
the restoration, this means that $x$ is the first $\GR$
vertex of the group during the restoration. This completes
our proof.
\end{proof}
The above lemma greatly simplifies our work, we just find
the first $\GR$ vertex $x$ such that $\SEG_{i}(x) = \SEG_{i}(v)$.
Once we have found  $x$, then we insert
$v$ in $\SSS_i$ by calculating all the relevant parameter
and then move on to process $x$. We now find the running
time of $\REE$. We list the  steps in this algorithm that dominates its running time.

\begin{enumerate}
\item Finding the $j$ for which $v$ is $\IH{j}$ (Step 6).

An easy (but sub-optimal space) solution for this problem will be to store this information for each vertex in an array, say $A$, of size $n$. However, the space required by $A$ will be $O(n \log(\log^* n))$ (as $2 \le j \le \log^* n)$. Since we do not have this much space, we use another strategy.

If $v$ is $\IH{2}$, then we can find it in $O(1)$ time. So, assume that $3 \le j \le \log^* n$. If $v$ is $\IH{j}$, then $\frac{m (\LL{j-1})^2}{n}<
\DD(v) \le \frac{m (\LL{j-2})^2}{n}$ or $(\LL{j-1})^2<
\frac{n\DD(v)}{m} \le (\LL{j-2})^2$. We make an array $A$ of size $O(\log^2n)$, such that each cell $k \in ((\LL{j-1})^{2}         ,(\LL{j-2})^{2}]$ has $A[k] = j$. Given any $v$, if the content of the cell $ \frac{n\DD(v)}{m}$
of $A$ is $j$, then    $v$ is $\IH{j}$. Since we  probe $A$ once, the time taken for this step is $O(1)$ time.

Note that the space  taken by the array $A$ is $O(\log^2 n \log(\log^*n))$ which is subsumed in the $O(n)$ notation.

\item Searching for $v.\NXT-1$ (the for loop inside the while loop (step 10-18))

Once we have found the starting vertex of the group (that is $k$) in the
while loop, the time taken in the for loop  is $O\Big(1+\frac{m}{n
\LL{i}}\Big)$. This is due to  Lemma \ref{lem:groupsize}
which states that the group size associated  with $v$ has
$1+ \frac{m}{n\LL{i}}$ vertices. 
\item Recoloring the  vertices (Step 28).

To this end, we should   maintain all the vertices that are colored $\WH$ by our restore algorithm and then enumerate them. Fortunately, there already exists a space-efficient data-structure that does this job.

\begin{lemma} \label{lem:enu}(Succinct Enumerate Dictionary
  \cite{BanerjeeC016})
  A set of elements from a universe of size n can be maintained
  using n + o(n) bits to support insert, delete, search and
  findany operations in constant time. We can  enumerate all
  elements of the set (in no particular order) in O(k +1)
  time where k is the number of elements in the set.
  \end{lemma}

We implement a enumerate dictionary in which we add all the vertices
that are colored $\WH$ by our restore algorithm. At the end of the while
loop of the restore algorithm, we use the enumerate dictionary to enumerate all
such vertices. We recolor each such vertex $\GR$ again and delete it from the
succinct dictionary. Using the above lemma, the extra space taken by the enumerate dictionary is $O(n)$.
\end{enumerate}

We now put everything together to calculate the total running time of $\REE(i)$.
Since, we restore vertices in  topmost
$\IS{i}$ only, we process only $\frac{n}{(\LL{i})^2}$ vertices
in the while loop of $\REE(i)$. Thus the while loop of $\REE(i)$ take $O\Big(\Big(1+\frac{m}{n \LL{i}}\Big) \times \frac{n}{(\LL{i})^2}\Big)$
time. Also, time taken by the recoloring step is proportional
to the number of vertices processed by $\REE(i)$, that is $O(\frac{n}{(\LL{i})^2})$.

In the restoration process, we use the data-structure
described
in Lemma \ref{lem:dict}
at most poly($n$) times, thus all insert and deletes are
successful
with probability $\ge 1-\frac{1}{n^{c}}$ where $c$ is some
constant. Thus, the algorithm succeeds with very high probability.

\begin{lemma}
\label{lem:restoreempty}
The time taken to restore $\SSS_i$ in $\REE(i)$ is $O\Big(\Big(1+\frac{m}{n \LL{i}}\Big) \times \frac{n}{(\LL{i})^2}\Big)$
with high probability.
\end{lemma}

\begin{algorithm}[h]
  \caption{ $\REF(i)$ }
  \label{}

  $(w,w.\NXT) \leftarrow \text{ top element in}\ \TT_i$\;

  $v \leftarrow \NN(w)[w.\NXT-1]$\;

  $counter \leftarrow 0$\;

  \Do{true}
  {

          $(\SEG_{i}(v),l_{v}) \leftarrow \SSS_{i+1}.\GET(v)$\;
        $k \leftarrow l_{v} \frac{\DD(v)}{\LL{i}} $\;

      \If{  $v$ is \IH{j} where $j \le i+1$}
      {
           $ l_{v} \leftarrow \HH_{j}.\GET(v)$\;

           $ k \leftarrow l_{v} \frac{\DD(v)}{\LL{j-2}}$ or $l_v$ (if $j=2)$\;

      }

      \For{$k' = k$ to
$k+1+ \frac{m}{n \LL{i}}$  }
      {
          $x \leftarrow \NN(u)[k']$\;
          \If { $x$ is $\GR$ and $x$ in present in $\SSS_{i+1}$}
          {
              $(\SEG_i(x),l_{x})\leftarrow \SSS_{i+1}.\GET(x)$\;

              \If{$\SEG_i(v) = \SEG_i(x)$}
           {
               break\;
           }
          }

      }

          $\SSS_i.\DL(v)$\;

        $\CC(v) \leftarrow \WH$\;
        \If{ $counter = \frac{n}{(\LL{i})^2}$}
        {
            Add $(v,k')$ on the top of stack $\TT_i$\;
            break;
           
        }

      $v \leftarrow x$\;

  }

  recolor all $\WH$ colored vertex during the above while
loop $\GR$ again\;
\end{algorithm}

Our last procedure $\REF(i)$ is called when $\SSS_i$ is full, that is, it contains vertices from the top two $\IS{i}$s of $\SSS$. The aim of $\REF(i)$ is to remove the vertices from the second top most $\IS{i}$ of $\SSS$. Thus, at the end of $\REF(i)$, $\SSS_i$ contains vertices of top $\IS{i}$ of $\SSS$. The procedure $\REF(\log^* n)$ is same as $\REE(\log^* n)$. For $i < \log^* n$, the procedure $\REF(i)$ is  similar to $\REE(i)$, we describe it next. 

We start with the top-most element of the trailer, say $w$.  Thus, the first vertex from the second topmost segment of $\SSS_i$ is $v \leftarrow \NN(w)[w.\NXT-1]$. Thus, we know that we have to delete $v$ from $\SSS_i$. However, before we delete $v$, we first find $\NN(v)[v.\NXT-1]$. The process to find this is same as done in $\REE(i)$. Then, we delete $v$ from $\SSS_i$ and set $v \leftarrow \NN(v)[v.\NXT-1]$. This process is carried out till we process all the vertices in the second topmost segment of $\SSS$. Thus, after our counter hits $\frac{n}{(\LL{i})^2}$, we have deleted all the vertices from the second topmost segment of $\SSS$ from $\SSS_i$.  Before we finish, we push the last processed vertex --- which is the trailer vertex of the second topmost segment of $\SSS$ ---  on top of trailer $\TT_i$.

The time taken by $\REF(i)$ is same as the time taken by $\REE(i)$. This is because the process to find $v.\NXT-1$ (given  $v$) is same for both the procedures. 
Also, the total number of vertices processed in both the procedures is same, that is $\frac{n}{(\LL{i})^2}$.
Thus, the time taken to restore
$\SSS_i$ in $\REF(i)$ is also $O\Big(\Big(1+\frac{m}{n \LL{i}}\Big) \times \frac{n}{(\LL{i})^2}\Big)$
with high probability.

\begin{lemma}
\label{lem:restorefull}
The time taken to restore $\SSS_i$ in $\REF(i)$ is $O\Big(\Big(1+\frac{m}{n
\LL{i}}\Big) \times \frac{n}{(\LL{i})^2}\Big)$
with high probability.
\end{lemma}

\section{Analysis}
\label{sec:analysis}
\subsection{Correctness of our Algorithm}
To prove the correctness, we just need to show our assumption
during the restoration procedure is true, that is $\SSS_{i+1}$
contains sufficient elements when $\SSS_{i}$ is restored.

\begin{lemma}
When $\SSS_{i}$ is restored, $\SSS_{i+1}$ contains at least top    $\frac{n}{2(\LL{i+1})^2}$
vertices of imaginary stack $\SSS$, where $1 \le i \le \log^* n$.
\end{lemma}

\begin{proof}
First, we note a crucial aspect of our algorithm. In 
Algorithm \ref{alg:maindfs}, 
$\IN$ or $\DL$ occurs in $\SSS_{i+1}$ before $\SSS_{i}$. 

We will now prove the lemma by induction on $i$ where $i$ decreases from $\log^* n$ to $1$. Let us first show
the base case, that is $\SSS_{\log^* n}$ always contains $\frac{n}{2\alpha^2}$ elements.  We have already seen that  $\SSS_{\log^* n}$
is correctly restored if it either becomes full or empty.
 So, $\SSS_{\log^* n}$ always contains top $\frac{n}{2\alpha^2}$
of the imaginary stack $\SSS$. 

Now, using induction hypothesis, we assume that all stacks $S_{j}$ ($ i+1\le j \le log^*n$) contains at least top    $\frac{n}{2(\LL{j})^2}$ vertices of the imaginary stack $S$. Now we will prove the statement of the lemma for stack $S_i$.

   We will use the fact that we $\IN$\ or $\DL$ in $\SSS_{i+1}$
before $\SSS_{i}$. Thus, whenever we are restoring $\SSS_{i}$,
(using induction hypothesis) $\SSS_{i+1}$  contains top $\frac{n}{2(\LL{i+1})^2}$
vertices of the imaginary stack $\SSS$.

In order to restore $\SSS_i$ correctly,
the only non-trivial requirement was that $\SSS_{i+1}$ contains
enough vertices during the restoration of $\SSS_i$.
Thus, we claim  $\SSS_i$ is always restored correctly. This completes the correctness of
the restore algorithm.
\end{proof}

\subsection{Space taken by our Algorithm}
We now calculate the space taken by  our algorithm. We list all our major data-structures and calculate their space.

\begin{enumerate}
\item  $\CC$ array

The $\CC$ array  is of size $n$ and each cell contains only three colors $\BL, \GR,$ or $\WH$. Thus each cell takes $2$ bits. Thus, the space taken by the $\CC$ array is $O(n)$.

\item Stack $\SSS_1$

$\SSS_1$ contains vertices of at most 2 segments of the imaginary segment $\SSS$. Thus, it contains at most $\frac{2n}{(\log n)^2}$ vertices. Each entry of the stack is of size $O(\log n)$. Thus, the total space taken by $\SSS_1$ is $O\Big(\frac{n}{\log n}\Big)$.

\item Dynamic Dictionary $\SSS_i$ ($2 \le i \le \log^* n$)

  Since
  $\SSS_i$ stores vertices from at most top two $\IS{i}$ of
  the imaginary stack $\SSS$, the number of vertices in $\SSS_i$
  is at most $\frac{2n}{(\LL{i})^2}$. In Section \ref{sec:info}, we saw that  the  information associated with
  each vertex of $\SSS_i$ is $O(\LL{i}$). Using Lemma \ref{lem:dict},
  the space taken by $\SSS_i$ is $\frac{2n}{(\LL{i})^{2}} \times
  \log \Big( \frac{n}{2n/\LL{i}}\Big) + \frac{2n}{(\LL{i})^{2}}
  \times \LL{i} = O\Big(\frac{n}{\LL{i}}\Big)$.  Thus, the cumulative  size  all $\SSS_i$'s is of  $ \sum_{i=2}^{\log^* n} O\Big(\frac{n}{\LL{i}}\Big) =O(n)$ bits.
  
  \item Trailers

   Since each $\IS{i}$ contains $\frac{n}{(\LL{i})^2}$ vertices, the total number of $\IS{i}$ is $O((\LL{i})^2)$. Thus, the number of elements in $\TT_i$ is $\le O((\LL{i})^2$).
  In each cell $\TT_i$, we explicitly store the entry $(v,v.\NXT)$.
  The total size of $\TT_i$ is thus $O(\log n (\LL{i})^{2})$ bits.
  The cumulative size of all $\TT_i$ is thus $O(\log n\sum_{i=1}^{\log^* n}
  (\LL{i})^{2})$ bits which are very small compared to our claimed
  space of $O(n)$ bits.
  
\item Dictionary for Heavy vertices, $\HH_i$ ($2 \le i \le \log^* n)$

We store the group number of $v$ in a   dynamic dictionary
  $\HH_i$, that is $\GRO_{i-2}(v.\NXT-1)$. Since we divide
  $\DD(v)$ into groups of size $\frac{\DD(v)}{(\LL{i-2})^{3}}$,
  the total number of groups is ($\LL{i-2})^{3}$.  This implies
  that total space required to represent the group number
  per cell in $\HH_i$ is $3\LL{i-1}$ bits. Also, by definition, each vertex in $\HH_i$ has degree $\ge \frac{m (\LL{i-1})^2}{n}$. Thus, the total number of vertices in $\HH_i$ can at most be $ O\Big(\frac{2n}{(\LL{i-1})^2}\Big)$. 

Using Lemma \ref{lem:dict},  the 
space
  taken for  $\HH_i$ is $\frac{2n}{(\LL{i-1})^{2}}
\times
  \log \Big( \frac{n}{2n/\LL{i-1}}\Big) + \frac{2n}{(\LL{i-1})^{2}}
  \times 3\LL{i-1} =O\Big(\frac{n}{\LL{i}}\Big)$.  Thus the cumulative 
size of all $\HH_i$'s is  $ \sum_{i=2}^{\log^* n} O\Big(\frac{n}{\LL{i}}\Big)
=O(n)$ bits.
\end{enumerate}

The reader can check that the total size of our algorithm is $O(n)$. We now find the total running time of our algorithm. 

\subsection{Running Time}
Using Lemma \ref{lem:maindfs}, we know that our main $\DFS$ algorithm (Algorithm \ref{alg:maindfs}) takes $O( m+ n\log^* n )$ time. The $n \log^*n$ term is due to the fact that we call $\IN$ and $\DL$ procedure at most $n \log^* n$ times in our algorithm. Except the restoration part, the $\IN$ and $\DL$ procedure takes $O(1)$ time (Lemma \ref{lem:ins_del_O1}). Thus the total running time of our algorithm (except the restoration procedure) is $O(m +n \log^* n)$.
To complete the analysis, we  need to find the total running time of our restore algorithm.

 Using Lemma  \ref{lem:restoreempty} and \ref{lem:restorefull},  the time taken to restore $\SSS_i$ ($2 \le i \le \log^* n)$  is  $O\Big(\Big(1+\frac{m}{n \LL{i-1}}\Big) \times \frac{n}{(\LL{i})^2}\Big)$.

We will count the number of times $\SSS_i$ is restored after it restored for the first time (this is just to simplify the analysis).
Whenever $\SSS_i$ is restored via $\REF$, take a look at
last time $\SSS_i$ was restored \footnote{This is the reason we left out the first restoration, as given any restoration we want to look back at the step when the previous restoration happened.}. At that point there were
 exactly$\ \frac{n}{(\LL{i})^2}$ elements in $\SSS_i$.  Thus, at least $\frac{n}{(\LL{i})^2}$
vertices must be  freshly added to $\SSS_i$. 
All these freshly added vertices must have changed their color
from $\WH$ to $\GR$. Since a vertex can change its color
from $\WH$ to $\GR$ only once in our $\DFS$ algorithm (when not processed in the restore procedure),
     $\SSS_i$ can be restored via $\REF$ at most $(\LL{i})^2$
times. Similarly, if $\SSS_i$ is restored via $\REE$, take
a look at the step at which it was restored previously.
At that time, $\SSS_i$ had exactly $\frac{n}{(\LL{i})^2}$. This implies
that at least  $\frac{n}{2(\LL{i})^2}$ have been deleted from $\SSS_i$. The only reason for deleting a vertex (when not processing it in a restore procedure) is that it has turned black. Since a vertex can
change its color from $\GR$ to $\BL$ only once in our $\DFS$ algorithm (when
not processed in the restore procedure), the total number of times $\SSS_i$ is
restored via $\REE$ is $O((\LL{i})^2)$. Thus, the total time
taken in restoring $\SSS_i$'s is as follows:

\begin{enumerate}
\item $i = \log^* n$

Remember that $\LL{\log^* n} = \alpha$ where $\alpha$ is some constant. Using Lemma \ref{lem:restorelast}, the time taken  to restore
$\SSS_{\log^* n}$
is $O(m+n)$. Thus the time taken for all  restorations of $\SSS_{\log^* n}$ is $O(\alpha^2(m+n)) = O(m+n)$.

\item $1\le i < \log^* n$

Using Lemma  \ref{lem:restoreempty} and \ref{lem:restorefull},
 the time taken to restore $\SSS_i$ ($2 \le i \le \log^*
n)$  is  $O\Big(\Big(1+\frac{m}{n \LL{i-1}}\Big) \times \frac{n}{(\LL{i})^2}\Big)$.
Thus the time taken for all the restoration
of $\SSS_i$ is $O\Big(\Big(1+\frac{m}{n \LL{i}}\Big) * \frac{n}{(\LL{i})^2} \times (\LL{i})^2 \Big) = O\Big(n+\frac{m}{ \LL{i}}\Big)$.  Hence, the total time taken to restore all $\SSS_i$'s ($1 \le i\le \log^* n$) 
is $O\sum_{i=1}^{\log^* n-1}\Big(n+ \frac{m}{\LL{i}} \Big) = O(n \log^* n+m)$. 
\end{enumerate}
Thus, the total time taken by our algorithm is $O(m + n \log^* n)$. This proves our main result, that is Theorem \ref{thm:main}.

\bibliographystyle{plain}
\bibliography{paper}
\end{document}